\documentclass[conference]{IEEEtran}
\IEEEoverridecommandlockouts
\usepackage{cite}
\usepackage{amsmath,amssymb,amsfonts}
\usepackage{algorithmic}
\usepackage{graphicx}
\usepackage{textcomp}
\usepackage{xcolor}
\usepackage{microtype}
\usepackage{bm}
\usepackage{amsthm}
\newcommand{\quotes}[1]{``#1''}
\newtheorem{theorem}{Theorem}
\newtheorem{lemma}{Lemma}
\newtheorem{d1}{Definition}

\newcommand{\h}{\hspace*{0.2in}}

\graphicspath{ {./figures/} }
\def\BibTeX{{\rm B\kern-.05em{\sc i\kern-.025em b}\kern-.08em
    T\kern-.1667em\lower.7ex\hbox{E}\kern-.125emX}}
\begin{document}

\title{Linearizable Replicated State Machines with Lattice Agreement\\
\thanks{Supported by CNS-1812349, NSF CNS-1563544, Huawei Inc., and the Cullen Trust for Higher Education Endowed Professorship.}
}

\author{\IEEEauthorblockN{1\textsuperscript{st} Xiong Zheng}
\IEEEauthorblockA{\textit{Electrical and Computer Engineering} \\
\textit{The University of Texas at Austin}\\
Austin, USA \\
zhengxiongtym@utexas.edu}
\and
\IEEEauthorblockN{2\textsuperscript{nd} Vijay K. Garg}
\IEEEauthorblockA{\textit{Electrical and Computer Engineering} \\
\textit{The University of Texas at Austin}\\
Austin, USA \\
garg@ece.utexas.edu}
\and
\IEEEauthorblockN{3\textsuperscript{rd} John Kaippallimalil}
\IEEEauthorblockA{\textit{Wireless Access Laboratories} \\
\textit{Huawei}\\
Plano, USA \\
John.Kaippallimalil@huawei.com}
}

 \maketitle

\begin{abstract}
This paper studies the lattice agreement problem in asynchronous systems and explores its application to building linearizable replicated state machines (RSM). First, we propose an algorithm to solve the lattice agreement problem in $O(\log f)$ asynchronous rounds, where $f$ is the number of crash failures that the system can tolerate. This is an exponential improvement over the previous best upper bound. Second, Faleiro et al have shown in [Faleiro et al. PODC, 2012] that combination of conflict-free data types and lattice agreement protocols can be applied to implement linearizable RSM. They give a Paxos style lattice agreement protocol, which can be adapted to implement linearizable RSM and guarantee that a command can be learned in at most $O(n)$ message delays, where $n$ is the number of proposers. Later on, Xiong et al in [Xiong et al. DISC, 2018] give a lattice agreement protocol which improves the $O(n)$ guarantee to be $O(f)$. However, neither protocols is practical for building a linearizable RSM. Thus, in the second part of the paper, we first give an improved protocol based on the one proposed by Xiong et al. Then, we implement a simple linearizable RSM using the our improved protocol and compare our implementation with an open source Java implementation of Paxos. Results show that better performance can be obtained by using lattice agreement based protocols to implement a linearizable RSM compared to traditional consensus based protocols.  
\end{abstract}


\section{Introduction}
Lattice agreement, introduced in \cite{b13}, to solve the atomic snapshot problem \cite{b20} in shared memory, is an important decision problem in distributed systems. In this problem, $n$ processes start with input values from a lattice and need to decide values which are comparable to each other in spite of $f$ process failures. 

There are two main applications of lattice agreement. First,
Attiya et al \cite{b13} give a $\log n$ rounds algorithm to solve the lattice agreement problem in synchronous message systems and use it as
a building block to solve the atomic snapshot object. 
Second, Faleiro et al\cite{b6} propose the problem of generalized lattice agreement
and demonstrate that the combination of conflict-free data types (CRDT) and generalized lattice agreement protocols can implement a special class of RSM and provides linearizability \cite{b23}. We call this class of state machines as Update-Query (UQ) state machines. The operations of UQ state machines can be classified into two kinds: updates (operations that modify the state) and queries or reads (operations that only return values and do not modify the state). An operation that both modifies the state and returns a value is not supported. Besides, all the operations are assumed to be deterministic. In this paper, when we talk about linearizable replicated state machine, we actually mean UQ state machine. We call such a replicated state machine built using lattice agreement as $LaRSM$. 

For the lattice agreement problem in asynchronous message system, Faleiro et al first presents a Paxos style protocol in \cite{b6} when a majority of processes are correct. Their algorithm needs $O(n)$ asynchronous round-trips in the worst case. They also propose a protocol for the generalized lattice agreement problem, adapted from their protocol for lattice agreement, which requires $O(n)$ message delays for a value to be learned. A procedure for building a linearizable RSM is also given, which requires $O(n)$ message delays for a request to be delivered, by combining CRDT and the their protocol for generalized lattice agreement. 
Later, an algorithm to solve the lattice agreement in asynchronous systems with $O(f)$ asynchronous round-trips was proposed by Xiong et al in \cite{b10}. They also propose a protocol for the generalized lattice agreement which improves the $O(n)$ message delays to $O(f)$. In this work, we improve the upper bound for the lattice agreement problem to be $O(\log f)$ by giving an algorithm.

Since lattice agreement can be applied to implement a linearizable RSM, if we can solve lattice agreement problem efficiently, we may not need consensus based protocol in some cases. From the theoretical perspective, using lattice agreement instead of consensus is promising, since lattice agreement has been shown to be a weaker decision problem than consensus. In synchronous message passing systems, consensus cannot be solved in fewer than $f + 1$ rounds \cite{b21}, but lattice agreement can be solved in $\log f + 1$ rounds \cite{b10}. In asynchronous systems, the consensus problem cannot be solved even with one failure \cite{b8}, whereas the lattice agreement problem can be solved in $O(\log f)$ when a majority of processes is correct. 

Replicated state machine \cite{b11} is a popular eager technique for fault tolerance in  a distributed system. Traditional replicated state machines typically enforce strong consistency among replicas by using a consensus based protocol to order all requests from the clients. In this approach, each replica executes all the request in an identical order to ensure that all replicas are at the same state at any given time. The most popular consensus based protocol for building a replicated state machine is Paxos\cite{b1,b2}. In the Paxos protocol, processes are divided into three different roles: proposer, acceptor and learner. Since the initial proposal of Paxos, many variants have been proposed. FastPaxos \cite{b5} reduces the typical three message delays in Paxos to two message delays by allowing clients to directly send commands to acceptors. MultiPaxos \cite{b24} is the typical deployment of Paxos in the industrial setting. It assumes that usually there is a stable leader which acts as a proposer, so there is no need for the first phase in the basic Paxos protocol. CheapPaxos \cite{b25} extends basic Paxos to reduce the requirement in the number of processors. Even though in the Paxos protocol, there could be multiple proposers, usually only one proposer is used in practice due to its non-termination problem when there are multiple proposers. Thus, all of them suffer from the performance bottleneck of a leader. Also, the unbalanced communication pattern limits the utilization of bandwidth available in all of the network links connecting the servers. SPaxos \cite{b26} is a Paxos variant which tries to offload the leader by disseminating clients to all replicas. However, the leader is still the only process which can order a request.
Although \cite{b6} has demonstrated that generalized lattice agreement protocol can be applied to implement a linearizable RSM, both the algorithms proposed in \cite{b6} and \cite{b10} are impractical for building a linearizable RSM, due to a problem we will explain in later section. Thus, we also propose an improved algorithm for the generalized lattice agreement problem. The improvements in the proposed algorithm are specifically designed to make it practical to build a linearizable RSM. 

In summary, this paper makes the following contributions:
\begin{itemize}
\item We present an algorithm, $AsyncLA$, to solve the lattice agreement in asynchronous system in $O(\log f)$ rounds, where $f$ is the number of maximum crash failures. This bound is an exponential improvement to the previously known best upper bound of $O(f)$ by \cite{b10}. 

\item We give an improved algorithm for the generalized lattice agreement protocol based on the one proposed in \cite{b10} to make it practical to implement a linearizable RSM. We also present optimizations for the procedure proposed in \cite{b6} to implement a linearizable RSM from a generalized lattice agreement protocol.

\item We implement a simple linearizable RSM in Java by combining a CRDT map data structure and our improved generalized lattice agreement algorithm. We demonstrate its performance by comparing with SPaxos. Our experiments show that LaRSM achieves around 1.3x times throughput than SPaxo and lower operation latency.  
\end{itemize}
\section{System Model and Problem Definitions}
\subsection{System Model}
We consider a distributed message passing system with $n$ processes, $p_1, \dots, p_n$, in a completely connected
topology. We only consider asynchronous systems, which means that there is no upper bound
on the time for a message to reach its destination. The model assumes that processes may
have crash failures but no Byzantine failures. The model parameter $f$ denotes the maximum number of processes that may crash in a run. We do not assume that the underlying communication system is reliable. The peer to peer network could be partitioned unpredictably. We need to build a replicated state machine which satisfy partition tolerance and provide as much availability and consistency as possible. 

\subsection{Lattice Agreement}
In the lattice agreement problem, each process $p_i$ can propose a value $x_i$ in a join semi-lattice ($X$, $\leq$, $sqcup$) and must decide on some output $y_i$ also in $X$. An algorithm solves the lattice agreement problem if the following properties are satisfied:

\textbf{Downward-Validity}: For all $i \in [1..n]$, $x_i \leq y_i$. 

\textbf{Upward-Validity}: For all $i \in [1..n]$, $y_i \leq \sqcup\{x_1,...,x_n\}$.

\textbf{Comparability}: For all $i \in [1..n]$ and $j 
\in [1..n]$, either $y_i \leq y_j$ or $y_j \leq y_i$. \\

The definition of height of a value and height of a lattice is given as below:
\begin{d1}
The height of a value $v$ in a lattice $X$ is the length of longest path from any minimal value to $v$, denoted as $h_X(v)$ or $h(v)$ when it is clear. 
\end{d1}

\begin{d1}
The height of a lattice $X$ is the height of its largest value, denoted as $h(X)$.
\end{d1}

\subsection{Generalized Lattice Agreement}
In the generalized lattice agreement problem, each process may receive a possibly infinite sequence of values belong to a lattice at any point of time. Let $x_i^{p}$ denote the $i$th value received by process $p$. The aim is for each process $p$ to learn a sequence of output values $y_j^p$ which satisfies the following conditions:

\textbf{Validity}: any learned value $y_j^p$ is a join of some set of received input values.

\textbf{Stability}: The value learned by any process $p$ is non-decreasing: $j < k \implies y_j^{p} \leq y_k^p$.

\textbf{Comparability}: Any two values $y_j^p$ and $y_k^q$ learned by any two process $p$ and $q$ are comparable. 

\textbf{Liveness}: Every value $x_i^p$ received by a correct process $p$ is eventually included in some learned value $y_k^q$ of every correct process $q$: i.e, $x_i^p \leq y_k^q$. \\

\section{Asynchronous Lattice Agreement in $O(\log f)$ Rounds}
In this section, we give an algorithm to solve the lattice agreement problem in asynchronous system which only needs $O(\log f)$ asynchronous rounds. The proposed algorithm is inspired by algorithms in \cite{b13} and \cite{b10}. The basic idea is to apply a \textit{Classifier} procedure to divide processes into \textit{master} and \textit{slave} groups and ensure that any process in the master group have values great than or equal to any process in the slave group. The \textit{Classifier} procedure is shown in Figure \ref{fig:classifier}. The main algorithm, {\em AsyncLA}, is shown in Figure \ref{fig:async_la}. Before we formally present the algorithm, we give some definitions.

\begin{d1}[label]
Each process has a $label$, which serves as a knowledge threshold and is passed as the threshold value $k$ whenever the process calls the \textit{Classifier} procedure.
\end{d1}

\begin{d1}[group]
A $group$ is a set of processes which have the same label. The label of a group is the label of the processes in this group. Two processes are said to be in the same group if and only if they have the same labels.
\end{d1}

Now, let us look at the {\em Classfier} procedure. Note that the main functionality of the {\em Classifier} is to divide the processes in the same group into two groups: the {\em master} group and the {\em slave} group and ensure that processes in {\em master} group have values greater than or equal to processes in \textit{slave} group. Details of the \textit{Classifier} procedure for $p_i$ are shown below:\\
\h Line 0: $p_i$ set its $acceptVal$ to be empty, which is used to store the set of $<value,label>$ pairs received from all processes.
Note that this includes values from processes that are not in the same group.\\
\h Line 1-2: $p_i$ sends a \textit{write} message containing its input value $v$ and the threshold value $k$ to all and wait for $n - f$ \textit{write\_ack}s. This step is to ensure the value and label of $p_i$ is in the $acceptVal$ set of $n - f$ processes. \\
\h Line 3-5: $p_i$ sends a \textit{read} message with its current round number $r$ to all processes and wait for $n - f$ \textit{read\_ack}s. It collects all the values associated with the same label $k$ in a set $U$, i.e, collects all values from processes within the same group. It may seem that line 3-5 are performing the same functionality as \textit{line} 1-2 and there is no need to have this part, since both are sending a message to all and waiting for $n - f$ acks. However, this part is actually the key of the \textit{Classifier} procedure. 
\\
\h Line 6-14: $p_i$ performs classification based on received values. Let $w$ be the join of all received values in $U$. If the height of $w$ in lattice $L$ is greater than the threshold value $k$, then $p_i$ sends a \textit{write} message with $w$,$k$ and $r$ to all and waits for $n - f$ \textit{write\_ack}s with round number $r$. Then in \textit{line} 10-12, it takes the join of $w$ and all the values contained in the \textit{write\_ack}s  from the same group denoted as $w'$. It returns $(w',master)$ as output of the \textit{Classifier} procedure in which \textit{master} indicates its classified into $master$ group in the next round. Otherwise, it returns its own input value $v$ and \textit{slave}. \\
\h In each round, when $p_i$ receives a $write$ message from some process $p_j$, it includes the associated value and label into its \textit{acceptVal} set. Then, it sends back a \textit{write\_ack} message with its current \textit{acceptVal} set to $p_j$. Similarly, when receiving a \textit{read} message from $p_j$, it sends back a \textit{read\_ack} message with its current \textit{acceptVal} set to $p_j$.  

Now we discuss the main algorithm \textit{AsyncLA}. The basic idea of \textit{AsyncLA} is to construct a binary tree of {\em Classifier}s and each process goes through this binary tree. After a process completes execution of one {\em Classifier} node, if it is classified as {\em master}, it goes to the right subtree, otherwise, it goes to the left subtree. Notice that after one round of exchanging values and taking joins, each process must know at least $n - f$ values. Since there are at most $n$ values, we set the threshold value of the root {\em Classifier} to be $\frac{(n - f) + n}{2} = n - \frac{f}{2}$. Thus, the height of the binary tree is $\log f$. Each process has a label, which is equal to the threshold value of the {\em Classifier} node it is currently invoking. So, the initial label for each process is $n - \frac{f}{2}$. Let $y_i$ denote the output value of $p_i$. The algorithm for $p_i$ proceeds in asynchronous rounds. Let $v_i^r$ denote its value at the beginning of round $r$.

At round 0, $p_i$ sends a \textit{value} message with its initial input $x_i$ to all and wait for $n - f$ \textit{value} messages for round 0 from all. It sets $v_i^1$ as the join of all values received in this round. This round is to make sure that height of the sublattice formed by all current values has height at most $f$. Then, at each round $r$ between 1 to $\log f$, $p_i$ invokes the \textit{Classifier} procedure with its current value $v_i^r$ and current label $l_i$ as input. Based on the output of the classifier, $p_i$ adjust its label by some value. If it is classified as master, then it increases its label by $\frac{f}{2^{r+1}}$,which is equals to the threshold value of the next {\em Classifier} it will invoke. Otherwise, it reduces its label by $\frac{f}{2^{r + 1}}$. At the end of round $\log f$, $p_i$ outputs $v_i^{\log f + 1}$ as its decision value.

\begin{figure}[htbp] 
\fbox{\begin{minipage}[t]  {3.3in}
\noindent
\underline{{\bf \textit{Classifier}$(v, k, r)$:}} \\
$v$: input value \h $k$: threshold value \h $r$: round number\\

{\bf 0:} {\em acceptVal} := $\emptyset$ // set of $<$value, label$>$ pairs.\\

/* write */\\
{\bf 1:} Send $write(v, k, r)$ to all\\
{\bf 2:} wait for $n - f$ $write\_ack(-, -, r)$\\

/* read */\\
{\bf 3:} Send $read(r)$ to all\\
{\bf 4:} {\bf wait} for $n - f$ $read\_ack(-, -, r)$\\
{\bf 5:} Let $U$ be values contained in received acks with label equals $k$\\

/* Classification */\\
{\bf 6:} Let $w := \sqcup\{u: u \in U\}$\\
{\bf 7: if} {$h(w) > k$} \\
{\bf 8:} \h Send $write(w, k, r)$ to all \\
{\bf 9:} \h wait for $n - f$ $write\_ack(-, -, r)$\\
{\bf 10:} \h Let $U'$ be values contained in received acks with label equals $k$\\
{\bf 11:} \h Let $w' := w \sqcup\{u: u \in U'\}$ \\
{\bf 12:} \h {\bf return} ($w'$, \textit{master})\\
{\bf 13:} {\bf else}\\
{\bf 14:} \h {\bf return} ($v$, \textit{slave})\\

\bf{Upon} receiving $write(v_j, k_j, r_j)$ from $p_j$\\
\h $acceptVal := acceptVal ~\cup~ <v_j, k_j>$\\
\h Send $write\_ack(acceptVal, r_j)$ to $p_j$\\

\bf{Upon} receiving $read(r_j)$ from $p_j$\\
\h Send $read\_ack(acceptVal, r_j)$ to $p_j$\\
\end{minipage} 
} 
\caption{\textit{Classifier} \label{fig:classifier}}
\vspace*{-.2in}
\end{figure}

\begin{figure} [htbp]
\fbox{\begin{minipage}[t]  {3.3in}
\noindent
\underline{$\bm{AsyncLA}(x_i)$ for $p_i$:} \\
$x_i$: input value \\
$y_i$: output value\\

 $v_i^{r}$ : value of $p_i$ at the beginning of round $r$ \\
 $l_i := n - \frac{f}{2}$ // initial label\\

/* Round 0 */ \\
Send {\em value($x_i$, 0)} to all \\
{\bf wait for} $n - f$ messages of form {\em value(-, 0)} \\
Let $U$ denote the set of all received values \\

/* Round 1 to $\log f$ */ \\
$v_i^1 := \sqcup \{ u ~|~ u \in U\}$ \\
 {\bf for} {$r := 1$ to $\log f$}\\
 \h ($v_i^{r + 1}, class$) := $Classifier(l_i, v_i^{r}, r)$\\
 \h {\bf if} {$class = master$} \\
 \h \h $l_i := l_i + \frac{f}{2^{r + 1}}$\\
 \h {\bf else} \\
 \h \h $l_i := l_i - \frac{f}{2^{r + 1}}$\\
 {\bf end for}\\
 $y_i := v_i^{\log f + 1}$
\end{minipage}
} 
\caption{Algorithm $AsyncLA$ \label{fig:async_la}}
\vspace*{-.2in}
\end{figure}

\subsection{Proof of Correctness}
We now prove the correctness of the proposed algorithm. Let $w_i^r$ be the value of $w$ at $line$ 6 of the \textit{Classifier} procedure at round $r$.

\begin{lemma} \label{lem:cls}
Let $G$ be a group at round $r$ with label $k$. Let $L$ and $R$ be two nonnegative integers such that $L \leq k \leq R$. If $L < h(v_i^{r}) \leq R$ for every process $i \in G$, and $h(\sqcup\{v_i^{r}: i \in G\}) \leq R$, then \\
(p1) for each process $i \in M(G)$, $k < h(v_i^{r + 1}) \leq R$\\
(p2) for each process $i \in S(G)$, $L < h(v_i^{r + 1}) \leq k$\\
(p3) $h(\sqcup\{v_i^{r + 1}: i \in M(G)\}) \leq R$ \\
(p4) $h(\sqcup\{v_i^{r + 1}: i \in S(G)\}) \leq k$, and\\
(p5) for each process $i \in M(G)$, $v_i^{r + 1} \geq \sqcup\{v_i^{r + 1}: i \in S(G)\}$   
\end{lemma}

\begin{proof}
\textbf{(p1)-(p3):} Immediate from the \textit{Classifier} procedure.\\
\textbf{(p4):} Proved by contradiction. Let us assume that $h(\sqcup\{v_i^{r + 1}: i \in S(G)\}) > k$. Since $v_i^{r + 1} = v_i^{r}$ for each process $i \in S(G)$, we have $h(\sqcup\{v_i^{r}: i \in S(G)\}) > k$. Let process $j$ be the last one in $S(G)$ to complete $line$ 2. When process $j$ starts executing $line$ 3, all other processes which are in $S(G)$ have already written their values to at least a majority of processes, that is, for any process $i \in S(G) \wedge i \not = j $, a majority of processes have included $<v_i, k>$ into their $acceptVal$ set. Thus, process $j$ would receive $<v_i, k>$ for any process $i \in S(G) \wedge i \not = j$, since any two majority of processes have at least one intersection. Then, we have $w_j^r = \sqcup\{v_i^{r}: i \in S(G)\}$. Thus, $h(w_j^r) = h(\sqcup\{v_i^{r}: i \in S(G)\}) > k$, which means $j \in M(G)$, a contradiction.\\
\textbf{(p5):} To prove $(p5)$, we need to show for any process $i \in M(G)$ and $j \in S(G)$, $v_i^{r + 1} \geq v_j^{r + 1} = v_j^{r}$. Let us consider the following three cases.\\
Case 1: when $i$ completes $line$ 9 and $j$ has not started $line$ 1. In this case, process $j$ would receive $<w_i^r, k>$ from at least one process at {\em line} 4, since any two majority of processes have at least one process in common. Then $j$ would be in $M(G)$ instead of $S(G)$, contradiction. \\
Case 2: when $j$ completes $line$ 2 and $i$ has not started $line$ 8. In this case, $i$ would receive $<v_j^r, k>$ from at least one process. Then, $v_i^{r + 1} \geq v_j^{r + 1}$.\\ 
Case 3: $i$ and $j$ are executing $line$ 1-2 and $line$ 8-9 concurrently. In this case, there exists a process $k$ which receives both $<w_i^r, k>$ and $<v_j^r, k>$. If $k$ receives $i$ first, then $j$ would receive $<w_i, k>$, contradiction. If $k$ receives $j$ first, then $i$ would receive $<v_j^r, k>$, which indicates  $v_i^{r + 1} \geq v_j^{r + 1}$. 
\end{proof}
Based on the above properties, we can have the following lemma. 

\begin{lemma}\label{lem:dec}
Let $G$ be a group of processes at round $r$ with label $k$. Then \\
(1) for each process $i \in G$, $k - \frac{f}{2^r} < h(v_i^r) \leq k + \frac{f}{2^r}$ \\
(2) $h(\sqcup\{v_i^r: i \in G\}) \leq k + \frac{f}{2^r}$
\begin{proof}
 By induction on round number $r$. When $r = 1$, label $k = n- \frac{f}{2}$, it is straightforward to have $n - f < h(v_i^r) \leq n$, since each process receives at least $n - f$ values and the height of input lattice is at most $n$. For the induction step, assume lemma \ref{lem:dec} holds for all groups at round $r - 1$. Consider an arbitrary group $G$ at round $r > 1$ with parameter $k$. Let $G'$ be the parent group of $G$ at round $r - 1$ with parameter $k'$. Consider the \textit{Classifier} procedure executed by all processes in $G'$ with parameter $k'$. By induction hypothesis, we have: \\  
 (1) for any process $i \in G'$, $k' - \frac{f}{2^{r - 1}} < h(v_i^{r - 1}) \leq k' + \frac{f}{2^{r - 1}}$ \\
 (2) $h(\sqcup\{v_i^{r - 1}: i \in G'\}) \leq k' + \frac{h}{2^{r - 1}}$.

 Let $L = k' - \frac{f}{2^{r - 1}}$ and $R = k' + \frac{f}{2^{r - 1}}$, then (1) and (2) are exactly the conditions of Lemma \ref{lem:cls}. Consider the following two cases: 

 Case 1: $G = M(G')$. Then $k = k' + \frac{f}{2^r}$. From (p1) and (p3) of Lemma \ref{lem:cls}, we have: \\
 (1) for any process $i \in G$, $k - \frac{f}{2^r} < h(v_i^r) \leq k + \frac{f}{2^r}$ \\
 (2) $h(\sqcup\{v_i^r: i \in G\}) \leq h(v_i^r) \leq k + \frac{f}{2^r}$. 
 
 Case 2: $G = S(G')$. Then $k = k' - \frac{f}{2^r}$. Similarly, from (p2) and (p4) of Lemma \ref{lem:cls}, we have the same equations. 
 \end{proof}
\end{lemma}

From Lemma \ref{lem:dec}, we directly have the following lemma. 
\begin{lemma}\label{lem:term}
Let $i$ and $j$ be two processes that are within the same group $G$ at the end of round $r = \log f + 1$. Then $v_i^{r + 1}$ and $v_j^{r + 1}$ are equal.
\begin{proof}
Let $G'$ be the parent of $G$ with parameter $k'$. Assume without loss of generality that $G = M(G')$. The proof for the case $G = S(G')$ follows in the same manner. Since $G'$ is a group at round $\log f$, by Lemma \ref{lem:dec}, we have: \\
 (1) for each process $p \in G'$, $k' - 1 < h(v_p^{\log f} ) \leq k' + 1$, and\\
 (2) $h(\sqcup\{v_p^{\log f}: p \in G'\}) \leq k' + 1$

 Since $i \in G'$ and $j \in G'$, (1) and (2) hold for both process $i$ and $j$. By the assumption that $G = M(G')$, at round $\log f$,  process $i$ and $j$ execute the \textit{Classifier} procedure with parameter $k'$ in group $G'$ and be classified as \textit{master} and proceed to group $G = M(G')$. Let $L = k' - 1$ and $R = k' + 1$, then by applying Lemma \ref{lem:cls}($p1$) we have $k' < h(v_i^{\log f + 1}) \leq k' + 1$ and $k' < h(v_j^{\log f + 1}) \leq k' + 1$, thus $h(v_i^{\log f + 1}) = h(v_j^{\log f + 1}) = k' + 1$. Similarly, by Lemma \ref{lem:cls}($p3$), we have $h(\sqcup\{v_i^{\log f + 1}, v_j^{\log f + 1}\}) = k' + 1$. Thus, $v_i^{\log f + 1} = v_j^{\log f + 1}$. Therefore, $v_i^{r}$ and $v_j^{r}$ are equal at the beginning of round $r = \log f + 1$.
\end{proof}
\end{lemma}

\begin{lemma}\label{lem:dom}
Let process $i$ decides on $y_i$. Let $G$ be a group at round $r$ such that $i \in S(G)$, then $y_i \leq \sqcup \{v_i^{r + 1}: i \in S(G)\}$. 
\begin{proof}
Immediate from $p2$ and $p4$ of Lemma \ref{lem:cls}. 
\end{proof}
\end{lemma}

\begin{lemma}\label{lem:comp}
Let $i$ and $j$ be any two processes in two different groups $G_i$ and $G_j$ at the end of round $\log f + 1$, then $y_i$ is comparable with $y_j$.
\begin{proof}
Since $G_i \not = G_j$, there must exist a group which contains both $i$ and $j$. Let $G$ be such a group with biggest round number $r$. Without loss of generality, assume $i \in S(G)$ and $j \in M(G)$. From Lemma \ref{lem:cls}($p5$), we have $v_j^(r + 1) \geq \sqcup\{v_i^{r + 1}: i \in S(G)\}$. From Lemma \ref{lem:dom}, we have $y_i \leq \sqcup \{v_i^{r + 1}: i \in S(G)\} \leq v_j^{r + 1}$. Note that the value held by any process is non-decreasing. Thus, $y_j \geq y_i$. Therefore, we have $y_i$ is comparable with $y_j$.
\end{proof}
\end{lemma}

Now, we have the main theorem. 
\begin{theorem}\label{theo:comp}
Algorithm {\em AsyncLA} solves the lattice agreement problem in $O(\log f)$ round-trips when a majority of processes is correct.
\begin{proof} 
\textit{Down-Validity} holds since the value held by each process is non-decreasing. \textit{Upward-Validity} follows because each learned value must be the join of a subset of all initial values which is at most $\sqcup \{x_1,...,x_n\}$. For \textit{Comparability}, from Lemma \ref{lem:term}, we know that any two processes which are in the same group at the end of {\em AsyncLA}, they must have equals values. For any two processes which are in two different groups, from Lemma \ref{lem:comp} we know they must have comparable values.
\end{proof}
\end{theorem}

\subsection{Complexity Analysis} 
Each invocation of the \textit{Classifier} procedure takes at most three round-trips. $\log f$ invocation of \textit{Classfier} results in at most $3 * \log f$ round-trips. Thus, the total time complexity is $3* \log f + 1$ round-trips. For message complexity, each process sends out at most 3 \textit{write} and \textit{read} messages and at most $3 * n$ \textit{write\_ack} and \textit{read\_ack} messages. Therefore, the message complexity for each process is $O(n * \log f)$. 

\section{Improved Generalized Lattice Agreement Protocol for RSM}
In this section, we give optimizations for the generalized lattice agreement protocol proposed in \cite{b10} to implement a linearizable RSM. The optimized protocol, $GLA_\Delta$, is shown in Fig \ref{fig:gla_alpha} with the two main changes marked using $\Delta$. Although we only have two primary changes compared to the original algorithm in \cite{b10}, we claim those changes are the key for its applicability in building a linearizable RSM. 

The basic idea of $GLA_\Delta$ is the same as the original algorithm in \cite{b10}. Each process invokes the {\bf Agree()} procedure, which is primarily composed of an execution of a lattice agreement instance to learn new commands. 
The {\bf Agree()} is automatically executed when the guard condition is satisfied. Inside the {\bf Agree()} procedure, a process first updates its {\em acceptVal} to be the join of current {\em acceptVal} and {\em buffVal}. Then, it starts a lattice agreement instance with next available sequence number. At each round of the lattice agreement, the process sends its current $acceptVal$ to all processes and waits for $n - f$ $ACK$s. If it receives any $decide$ $ACK$, it decides on the join of all $decide$ values. If it receives a majority of $accept$ $ACKs$, it decides on its current value. Otherwise, it updates its $acceptVal$ to be the join of all received values and starts next round. When a process receives a proposal from some other process, if the proposal is associated with a smaller sequence number, then it sends $decide$ $ACKs$ back with its decided value for that sequence number and includes the received value into its own buffer set. Otherwise, it waits until its current sequence number to be equals to the sequence number associated with the proposal. Then, it checks whether the proposed value contains its current $acceptVal$. If true, the process sends back a $accept$ $ACK$. Otherwise, it sends back a $reject$ $ACK$ along with its current $acceptVal$. When a process completes a lattice agreement instance for sequence number $s$, it stores decided values into $LV[s]$. Then it removes all learned values for sequence number $s - 1$. 

\begin{figure}[htbp]\begin{centering}
\fbox{\begin{minipage} {0.46\textwidth}
\underline{$\bm{GLA_\Delta}$ for $p_i$}\\
\h \textit{s} := 0 // sequence number\\
\h \textit{maxSeq} := -1 // largest sequence number seen\\
\h \textit{buffVal} := $\bot$ // commands buffer\\
\h \textit{LV} := $\bot$ // map from seq to learned commands set\\
\h \textit{acceptVal} := $\bot$ // current accepted commands set \\
\h \textit{active} := \textit{false} //proposing status\\

\textbf{Procedure Agree():}\\
  {\bf guard:} (\textit{active} = $false$) $\wedge$ (\textit{buffVal} $\neq \bot$ $\vee ~maxSeq \geq s$) \\
  {\bf effect:} \\
\h \textit{active} := $true$\\
\h \textit{acceptVal} := \textit{buffVal} $\sqcup$ \textit{acceptVal} \\
\h \textit{buffVal} := $\bot$\\
\h\\
\h /* Lattice Agreement with sequence number $s$ */\\
\h {\bf for $r := 1$ to $f + 1$} \\
\h\h \textit{val} := \textit{acceptVal}\\
\h\h  Send \textit{prop}(\textit{val}, $r, s$) to all\\
\h\h  {\bf wait for} $n - f$ \textit{ACK}($-, -, r, s$)\\
\h\h let $V$ be values in \textit{reject ACKs} \\
\h\h let $D$ be values in \textit{decide ACKs} \\
\h\h let \textit{tally} be number of \textit{accept ACKs} \\
\h\h {\bf if} {$|D| > 0$}\\
\h\h\h $val := \sqcup \{d ~| ~d \in D\}$ \\
\h\h\h {\bf break}\\
\h\h {\bf else if} {\textit{tally} $> \frac{n}{2}$} \\
\h\h\h {\bf break}\\
\h\h {\bf else}\\
\h\h \h Let $tmp$ := $\sqcup \{v ~| ~v \in V\}$ \\
\h\h\h \textit{acceptVal} := \textit{acceptVal} $\sqcup ~tmp$ \\
\h {\bf end for}\\
\\
\h \textit{LV}$[s]$ := \textit{val} \\
\h  \textit{acceptVal} := \textit{acceptVal} - $LV[s-1]$ \hspace{0.5in} $\Delta_1$ \\
\h \textit{s} := \textit{s} + 1\\
\h \textit{active} := \textit{false}\\

\textbf{on receiving} \textit{ReceiveValue}($v$):\\
\h \textit{buffVal} := \textit{buffVal} $ \sqcup ~v$\\

{\bf on receiving} \textit{prop}$(v_j, r, s')$ from $p_j$:\\
\h {\bf if} {$s' < s$}\\
 \h \h \textit{buffVal} := \textit{buffVal} $ \sqcup ~v_j$ \hspace{1in} $\Delta_2$\\
\h\h Send \textit{ACK}(\textit{\quotes{decide}, LV}[$s'$], $r, s'$)\\
\h\h {\bf return}\\
\h  \textit{maxSeq} := $\max \{s', \textit{maxSeq}\}$\\
\h {\bf wait until} $s' = s$ \\
\h {\bf if} {\textit{acceptVal} $\subseteq v_j$}\\
\h\h Send \textit{ACK}(\textit{\quotes{accept},} $-, r, s'$) \\
\h\h \textit{acceptVal} $:= v_j$ \\
\h {\bf else}  \\
\h\h Send \textit{ACK}(\textit{\quotes{reject}}, \textit{acceptVal}, $r, s'$)\\
\end{minipage}
}
\end{centering}
\caption{Algorithm $GLA_\Delta$ \label{fig:gla_alpha}}
\vspace*{-0.2in}  
\end{figure}

\subsection{Truncate the Accept and Learned Command Set}
Let us first look at the challenges of directly applying the generalized lattice agreement protocol in \cite{b10} or the one in \cite{b6} to implement a linearizable RSM. In a replicated state machine, each input value is a command from a client. Thus, the input lattice is a finite boolean lattice formed by the set of all possible commands. The order in this lattice is defined by set inclusion, and the join is defined as the union of two sets. This boolean input lattice poses a challenge for both the algorithms in \cite{b6} and \cite{b10}. In these algorithms, for each process (each acceptor process in \cite{b6}) there is an accept value set, which stores the join of whatever value the process has accepted. Now since the join is defined as union in the RSM setting, this set keeps increasing. For example, in Fig. \ref{fig:accProblem}, $p_1$, $p_2$ and $p_3$ first receive commands $\{a\}$, $\{b\}$ and $\{c\}$, respectively. They start the lattice agreement instance with sequence number $0$ and learn $\{a\}$, $\{a, b\}$ and $\{a, b, c\}$ respectively for sequence number 0. After that, $p_1$, $p_2$ and $p_3$ receive $\{d\}$, $\{e\}$, and $\{f\}$ as input, respectively. Now, they start a lattice agreement instance with the sequence number 1. In order to ensure comparability and stability of generalized lattice agreement, the learned command set and accept command set for sequence number 1 have to include the largest learned value of sequence 0, which is $\{a, b, c\}$, although each process only proposes a single command. Therefore, the accept and learned value set keeps increasing. This problem makes applying lattice agreement to implement a linearizable RSM impractical. 

\begin{figure}[htbp]
\centerline{\includegraphics[width=3in]{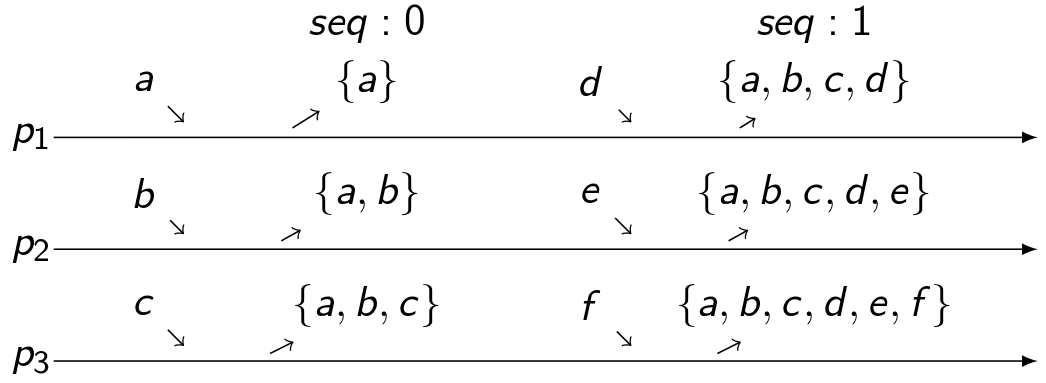}}
\caption{The Accept and Learned Value Set Keeps Increasing\label{fig:accProblem}}
\end{figure}

To tackle the always growing accept command set problem, we would like to have some way to truncate this set. A naive way is to remove all learned commands in the accept command set when proposing for the next available sequence number. This way does not work. Suppose we have two processes: $p_1$, $p_2$ and $p_3$. They propose $\{a\}$, $\{b\}$ and $\{c\}$, respectively for sequence number 0. After execution of lattice agreement for sequence number 0, suppose $p_1$, $p_2$ and $p_3$ both have learned value set and accept value set to be $\{a\}$, $\{a, b, c\}$, and $\{a, b, c\}$, respectively. It is easy to verify this case is possible for an execution of lattice agreement. When completing sequence number 0, all processes remove learned value set for sequence number 0 from their accept value set. Thus, the accept value set of all the three processes becomes to be empty. Now, if $p_1$, $p_2$ and $p_3$ start to propose for sequence number 1 with new commands $\{d\}$, $\{e\}$ and $\{f\}$. Since the accept command sets of $p_2$ and $p_3$ do not contain value $\{b\}$ and $\{c\}$, $p_1$ will never be able to learn $\{b\}$ and $\{c\}$. Thus, learned command set of $p_1$ for sequence 1 and the learned command set of $p_2$ and $p_3$ for sequence 0 are incomparable. Thus, we cannot remove all learned value set from the accept value set. Instead of removing all learned commands from the accept command set, we propose to remove all learned commands for the sequence numbers smaller than the largest learned sequence number from the accepted command set. In order to achieve this, the line marked by $\Delta_1$ in the pseudocode is added, compared to the original algorithm in \cite{b10}. In this line, after a process has learned a value set for sequence number $s$, it removes the learned value set corresponding to sequence number $s-1$ from its accept value set.

Second, as the state machine keeps running, the mapping of sequence number to learned commands, $LV$,  also keeps growing. Thus, we propose the following technique to truncate this map. Let each process record the largest sequence number for which all replicas have started proposing, denoted as $min\_seq$. Thus, all replicas have learned commands for any sequence number smaller than $min\_seq$, since each replica has to learn commands for each sequence. Besides, each replica also record the largest sequence number for which the corresponding learned values have been applied into state (executed), denote as $executed\_seq$. Then, each replica removes all learned commands in $LV$ with sequence number smaller than min of $min\_seq$ and $executed\_seq$. In this way, the learned commands map can be kept small. Since this improvement is trivial, we do not include it in the algorithm pseudocode.  

\subsection{Remove Forwarding}
In both the algorithms of \cite{b6} and \cite{b10}, a process has to forward all commands it receives to all other processes or proposers to ensure liveness. This forwarding results in load that is multiplied many fold, since many processes may propose the same request. We claim that this blind forwarding is a waste. In \cite{b10}, this forwarding is to ensure that the commands proposed by slow processes can also be learned. However, for the fast processes, there is no need to forward their requests to others because they can learn requests quickly. Therefore, instead of forwarding every request to all servers, we require that when a process receives some proposal with smaller sequence number than its current sequence number, it sends back a $decide$ message and also include the received proposal value into its own buffer set. These values will be proposed by the server in its next sequence number. In this way, only when a process is slow, its value will be proposed by the fast processes. This change is shown as addition of the line marked by $\Delta_2$ in the algorithm.

\subsection{Proof of Correctness} 
In this section, we prove the correctness of algorithm $GLA_\Delta$. Although we only have two primary changes compared to the algorithm in \cite{b10}, the correctness proof is quite different. Let $LearnedVal_s^p$ denotes the learned value of process $p$ after completing lattice agreement for sequence number $s$. Thus, $LearnedVal_s^i = \sqcup\{LV[t]: t \in [0...s]\}$. Let $accept_s^p$ denotes the value of $acceptVal$ of process $p$ at the end of sequence number $s$. 

The following lemma follows immediately from the {\em Comparability} requirement of the lattice agreement problem. 

\begin{lemma}\label{lem:comp_same_sequence}
For any sequence number $s$, \textit{$LV_p[s]$} is comparable with \textit{$LV_q[s]$} for any two processes $p$ and $q$.
\end{lemma}

The following lemma shows {\em Stability}. 
\begin{lemma}\label{lem:stability}
For any sequence number $s$, $LearnedVal_p^s$ $\subseteq$ $LearnedVal_q^{s + 1}$ for any two correct processes $p$ and $q$.
\begin{proof}
Proof by induction on sequence number $s$. 

The base case, $s = 0$. When $p$ completes sequence number 0, $LV_p[0]$ must be accepted by a majority of processes. That is, there exists a majority of processes which include  $LV_p[0]$ into their accept command set, i.e, into $acceptVal$. During the $q's$ execution of lattice agreement 1, it must learn $LV_p[0]$ because any two majority of processes have at least one common process. Thus, $LV_p[0] \subseteq LV_q[1]$. So, we have \textit{$LearnedVal_p^0$} $\subseteq$ \textit{$LearnedVal_q^{1}$}. 

The induction case. Assume that for sequence number $s$, we have \textit{$LearnedVal_p^s$} $\subseteq$ \textit{$LearnedVal_q^{s + 1}$} for any two processes $p$ and $q$. We need to show that \textit{$LearnedVal_p^{s+1}$} $\subseteq$ \textit{$LearnedVal_q^{s + 2}$}. Equivalently, we show that $LearnedVal_p^{s}  \cup LV_p[s + 1]$  $\subseteq$ $LearnedVal_q^{s + 1} \cup LV_q[s + 2]$. Thus, we only need to show that  $LV_p[s + 1]$  $\subseteq$ $LearnedVal_q^{s + 1} \cup LV_q[s + 2]$, since we have $LearnedVal_p^{s} \subset LearnedVal_q^{s + 1}$ by assumption. Consider any $v \in LV_p[s + 1]$. During $p's$ execution of lattice agreement for sequence number $s + 1$, $v$ must be included into $acceptVal$ by a majority of processes. Let $Q$ denotes such a majority of processes. Due to the change marked by $\Delta_1$, there could exist some process $j \in Q$ such that $v \not \in acceptVal_j^{s + 1}$. In this case, we must have $v \in LV_j[s] \subseteq LearnedVal_j^{s} \subseteq LearnedVal_q^{s + 1}$. In the other case, if $\forall j \in Q$, we have $v \in acceptVal_j^{s + 1}$. Then during $q's$ execution of lattice agreement for sequence number $s + 2$, $q$ must learn $v$ since $v$ is contained in the $acceptVal$ of a majority of processes. Thus, $v \in LV_q[s + 2]$. So, $\forall v \in LV_p[s + 1]$, we either have $v \in LearnedVal_q^{s + 1}$ or $v \in LV_q[s + 1]$. Therefore, we have $LV_p[s + 1]$  $\subseteq$ $LearnedVal_q^{s + 1} \cup LV_q[s + 2]$, which yields $LearnedVal_p^s \subseteq LearnedVal_q^{s + 1}$ for any two processes $p$ and $q$.
\end{proof}
\end{lemma}
Now, let us prove {\em Comparability}.
\begin{lemma}\label{lem:comparability}
For any sequence number $s$ and $s'$, $LearnedVal_p^s$ and $LearnedVal_q^{s'}$ are comparable for any two correct processes $p$ and $q$.
\begin{proof}
For $s' > s$ or $s' < s$, Lemma \ref{lem:stability} gives the result. So, we only need to consider the case $s = s'$. We prove this case by induction on sequence number $s$. 

The base case $s = 0$ immediately follows from Lemma \ref{lem:comp_same_sequence}. 

For the induction case, assume for sequence number $s$, $LearnedVal_p^s$ and $LearnedVal_q^{s}$ are comparable for any two processes $p$ and $q$. Need to show $LearnedVal_p^{s + 1}$ and $LearnedVal_q^{s + 1}$ are comparable. Equivalently, we can show $LearnedVal_p^{s} \cup LV_p[s + 1]$ and $LearnedVal_q^{s} \cup LV_q[s + 1]$ are comparable. Without loss of generality, assume $LearnedVal_p^{s} \subseteq LearnedVal_q^{s}$, the proof for the other case is similar. Let us consider the following two cases.\\
Case 1: $LV_p[s + 1] \subseteq LV_q[s + 1]$. By the assumption, we have $LearnedVal_p^{s} \cup LV_p[s + 1] \subseteq LearnedVal_q^{s} \cup LV_q[s + 1]$. \\
Case 2: $LV_q[s + 1] \subset LV_p[s + 1]$. From Lemma \ref{lem:stability}, we have $LearnedVal_q^{s} \subseteq LearnedVal_p^{s + 1} = LearnedVal_p^{s} \cup LV_p[s + 1]$. Therefore, $LearnedVal_q^{s} \cup LV_q[s + 1] \subseteq LearnedVal_p^{s} \cup LV_p[s + 1]$.
\end{proof}
\end{lemma}

\begin{theorem}\label{theo:corretness_gla_alpha}
Algorithm $GLA_\Delta$ solves the generalized lattice agreement problem when a majority of processes is correct.
\begin{proof}
\textit{Validity} holds since any learned value is the join of a subset of values received. \textit{Stability} follows from Lemma \ref{lem:stability}. \textit{Comparability} follows from Lemma \ref{lem:comparability}. \textit{Liveness} follows from the termination of lattice agreement. 
\end{proof}
\end{theorem}

\section{Improve the Procedure for Implementing a Linearizable RSM}
The paper \cite{b6} gives a procedure to implement a linearizable RSM by combining CRDT and a protocol for the generalized lattice agreement problem. The basic idea in \cite{b6} is to treat reads and writes separately. For a write command, say $cmd_w$, the receiving proposer invokes a lattice agreement instance with this write operation as input value and then wait until $cmd_w$ is included into its learned commands set (The learned command set stores all learned commands received from learners). Then, it returns response for $cmd_w$. For a read command, say $cmd_r$, the receiving proposer creates a null command, which is a command that has no effect. It invokes a lattice agreement instance with this null command and waits until its command is in the learned commands set. Then, it executes all commands stored in the learned command set and returns the response for $cmd_r$. In this paper, we propose some simple optimizations for this procedure. 

To tackle the aforementioned problems, we present the following two optimizations for the linearizable SMR procedure proposed in \cite{b6}. 

\subsection{Reduce Burden of Read}
In the procedure proposed in \cite{b6}, the learned commands are only executed when there is a read command and a read command can only return when the server completes executing all current learned commands. This results in high latency of a read operation. In order to reduce the latency of read operation and balance between reads and writes, each server applies newly learned commands whenever it completes a sequence number. 

Besides, for each read command, before returning a response, a null operation needs to be created and learned. This is not necessary. We only need to create one null operation for all read operations in the commands buffer and all those reads can be executed when that single null operation is learned.

\subsection{Remove Reads from Input Lattice}
In procedure proposed in \cite{b6}, the input lattice is formed by all update commands and all null commands, which is not necessary. The null commands are actually read commands. Since only updates change the state of the server and reads do not, only the lattice formed by all updates need to be considered. In the lattice agreement protocol, a basic and highly frequent operation for a process is to check whether a received proposal value, i.e, a set of commands, contains its current accept command set. Since we only need to consider the lattice formed by all the updates, a process only needs to check whether the subset of updates in the proposed command set contains the subset of updates in its current accept command set. 

\section{LaRSM vs Paxos}
In this section, we compare LaRSM and Paxos from both theoretical and engineering perspective. Table \ref{tab:LaRSM_Paxos} shows the theoretical perspective. For the engineering perspective, since there is no termination guarantee when multiple proposers exist in the system, Paxos is typically  deployed with only one single proposer (the leader). Only the leader can handle handle requests from the clients. Thus, in a typical deployment the leader the leader becomes the bottleneck and the throughput of the system is limited by the leader's resources. Besides, the unbalanced communication pattern limits the utilization of bandwidth available in all of the network links connecting the servers. However, there can be multiple proposers in LaRSM since termination is guaranteed. Multiple proposers can simultaneously handle requests from clients, which may yield better throughput. In failure case, new leader needs to be elected in Paxos and there could be multiple leaders in the system. During this time, the protocol may not terminate because of conflicting proposals. Even though there are ways to reduce conflicting proposals, generally it needs more rounds to learn a command when there are multiple leaders. However, a failure of a replica in LaRSM has limited impact on the whole system. This is because other replicas can still handle requests from clients as long as less than a majority of replicas has failed. 
In a typical deployment of Paxos, {\em pipelining} \cite{b1} is often applied to increase the throughput of the system.  In {\em pipelining}, the leader can concurrently issue multiple proposals. In LaRSM, however, there can be at most one proposal for each replica at any given time, because the $Stability$ and $Comparability$ of generalized lattice agreement require that next proposal can be issued only when current proposal terminates. Thus, LaRSM does not support {\em pipelining}. 

In summary, compared with Paxos, the main advantage of LaRSM is that it can have multiple proposers concurrently handling requests and the main disadvantage is that it does not support {\em pipelining} for each proposer.

\section{Evaluation}
In this section, we evaluate the performance of LaRSM and compare with SPaxos. 
Although the lattice agreement protocol proposed in this paper has round complexity of $O(\log f)$, it has large constant, which is only advantageous when the number of processes is large. In real case, the number of replicas is usually small, often 3 to 5 nodes. Thus, instead of using the lattice agreement protocol proposed in this paper, we use the lattice agreement protocol from \cite{b10} which runs in $f + 1$ asynchronous round-trips in our implementation. In order to evaluate LaRSM, we implemented a simple replicated state machine which stores a Java hash map data structure. We implement the hash map date structure to be a CRDT by assigning a timestamp to each update operation and maintain the last writer wins semantics. We measure the performance of SPaxos and our implementation in the following three perspectives: performance in the normal case (no crash failure), performance in failure case, and performance under different work loads. 

All the experiments are performed in Amazon’s EC2 infrastructure with micro instances. The micro instance has variable ECUs (EC2 Compute Unit), 1 vCPUs, 1 GBytes memory, and low to moderate network performance. All servers run Ubuntu Server 16.04 LTS (HVM) and the socket buffer sizes are equal to 16 MBytes. All experiments are performed in a LAN environment with all processes distributed among the following three availability zones: US-West-2a, US-West-2b and US-West-2c. 

 The keys and values of the map are string type. We limit the range of keys to be within 0 to 1000. Two operations are support: update and get. The update operation changes the value of a specific key. The get operation returns the value for a specific key. A client execute one request per time and only starts executing next request when it completes the first one. The request size is 20 bytes. For each request, the server returns a response to indicates its completeness. In order to compare with SPaxos, we set its crash model to be CrashStop. In this model, SPaxos would not write records into stable storage. In SPaxos, batching and pipelining are implemented to increase the performance of Paxos. There are some parameters related to those two modules: the batch size, batch waiting timeout and the window size. The batch size controls how many requests the batcher needs to wait before starting proposing for a batch. The batch waiting timeout controls the maximum time the batch can wait for a batch. The window size is the maximum number of parallel proposals ongoing. We set the batch size to be 64KB, which is the largest message size in a typical system. We set the batch timeout according to the number of clients from 0 to 10 at most. The window size is set to 2 as we found that increasing the window size further does not increase the performance in our evaluation.

\subsection{Performance in Normal Case}
In this experiment, we build a replicated state machine system with three instances. We test the throughput of the system and latency of operations while keep increasing the number of requesting clients. The load from the clients are composed of 50\% writes and 50\% reads. Figure \ref{fig:clients_throught} shows the throughput change of SPaxos and LaRSM. The throughput is measured by the number of requests handled per second by the system. The latency is the average time in milliseconds taken by the clients to complete execution of a request. We can see from Fig \ref{fig:clients_throught}, as we increase the number of requesting clients, the throughput of both SPaxos and LaRSM increase until there are around 1000 clients. At that point, the system reaches its maximum handling capability. If we further increase the clients number, the throughput of both LaRSM and SPaxos does not change in a certain range and begins to decrease if there are more requesting clients. This is because both systems do not limit the number of connections from the client side. A large number of clients connection results in large burden on IO, decreasing the system performance. Comparing SPaxos and LaRSM, we can see that LaRSM always has better throughput than SPaxos. The maximum gap is around 10000 requests/sec.

Figure \ref{fig:clients_latency} shows the latency change as the number of clients increases. In both LaRSM and SPaxos, read and write perform the same procedure, thus their latency should be same. So, in our evaluation, we just say operation latency. From Figure \ref{fig:clients_latency}, we find that operation latency of LaRSM is always increasing. As we increase the number of clients, the latency of SPaxos decreases first up to some point and then begins to increase. This performance is due to the fact that the latency of the average response time of all clients and SPaxos has a batching module which batches multiple requests from different clients to propose in a single proposal. Therefore, initially when there are very few clients, they can only propose a small number of requests in a single proposal, which makes the latency relatively higher. While the number of clients increases, more requests can be proposed in one single batch, thus the average latency for one client is decreased. Later on, if the number of clients increases further, the handling capability limit of the system increases the operation latency. Comparing SPaxos and LaRSM, we find that the latency of LaRSM is always around 5ms smaller.
\begin{figure}[htb]
\includegraphics[width=3.5in]{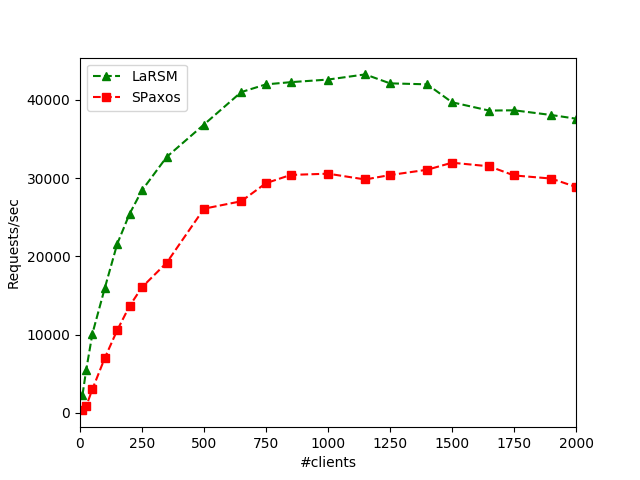}
\centering
\caption{Throughput of LaRSM and SPaxos with increasing number of clients\label{fig:clients_throught}}
\end{figure}

\begin{figure}[ht]
\includegraphics[width=3.5in]{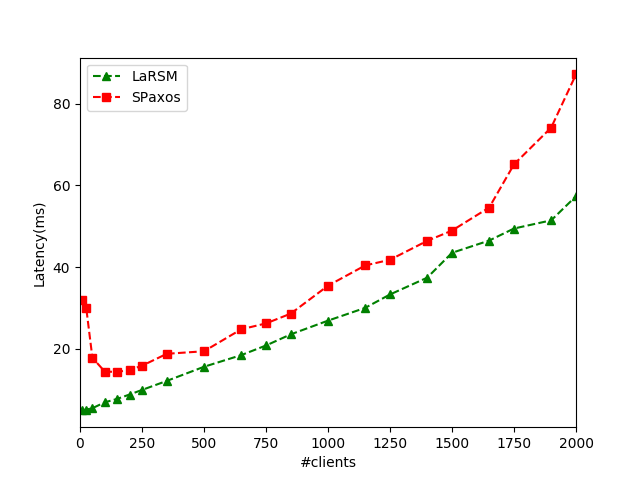}
\centering
\caption{Latency of LaRSM and SPaxos with Increasing Number of Clients\label{fig:clients_latency}}
\end{figure}

\subsection{Performance in Failure Case}
In this section, we evaluate the performance of both LaRSM and SPaxos in the case of failure. In this experiment, the replicated state machine system is composed of five replicas. There are 100 clients that keep issuing requests to the system. In LaRSM, since all replicas perform the same role and can handle requests from the clients concurrently. Thus, for loading balancing, each client randomly selects a replica to connect. Each client has a timeout, unlike SPaxos, this timeout is typically small. Timeout on an operation does not necessarily mean failure of the connected replica. It might also due to overload of the replica. In this case, the client randomly chooses another replica to connect. However, in SPaxos, the timeout set for a client is usually used to suspect the leader. That is, when an operation times out, most likely the leader has failed. Thus, the timeout in SPaxos is typically large. 

We run the simulation for 40 seconds. The first 10 seconds is for the system to warm up, so we do not record the throughput and latency data. A crash failure is triggered at 25th second after the start of the system. For LaRSM, we randomly shut down one replica since all replicas are performing the same role. For SPaxos, we shut down the leader, since crash of a follower does not have much impact on the system. Figure \ref{fig:failure_throughput} shows the throughput of both LaRSM and SPaxos. Figure \ref{fig:failure_latency} shows the latency change. From Figure \ref{fig:failure_throughput} and Figure \ref{fig:failure_latency}, for LaRSM we can see that when the failure occurs, the throughput drops sharply from around 20K requests/sec to around 15K requests/sec, but not to 0. However, the throughput of SPaxos drops to zero when leader fails. The latency of LaRSM only increases slightly, whereas the latency of SPaxos goes to infinity (Note that in the figure it is shown as around 500ms). This is because  when leader fails, SPaxos stops ordering requests, thus no requests are handled by the system. For LaRSM, the clients which are connected to the failed replica, would have timeout on their current requests and then randomly connect to another replica. As discussed before, this timeout is usually much smaller than the timeout for suspecting a failure in SPaxos. Thus, the latency of a client in LaRSM only increases by a small amount. After the failure, the throughput of LaRSM remains around 16K requests/sec, which is because now there is one less replica in the system and the handling capability of the system decreases. For SPaxos, after a new leader is selected, the throughput increases to be a level slightly smaller than the throughput before the failure and the latency also decreases to be slightly higher than the latency before the failure. We also find that even though the throughput of LaRSM drops when a failure occurs, it still has better throughput than SPaxos, which indicates the good performance of LaRSM.

\begin{figure}[t]
\includegraphics[width=3.5in]{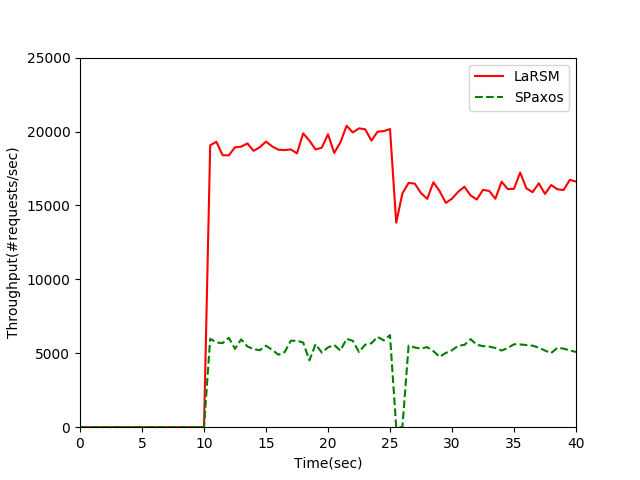}
\centering
\caption{Throughput in Case of Failure\label{fig:failure_throughput}}
\end{figure}

\begin{figure}[htb]
\includegraphics[width=3.5in]{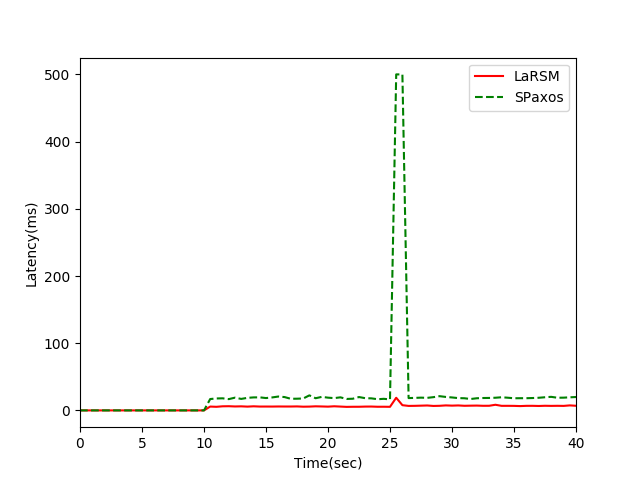}
\centering
\caption{Latency in Case of Failure\label{fig:failure_latency}}
\end{figure}

\subsection{Performance under Different Loads}
In this part, we evaluate the performance of LaRSM on different types of work loads. This evaluation is done in a system of three replicas with 500 clients keep issuing requests. We measure the throughput and latency as we increase the ratio of reads in a work load. Figure \ref{fig:loads_throught} and Figure \ref{fig:loads_latency} give the throughput and latency change respectively. It is shown in those two figures that as the ratio of reads increases in a work load, the throughput of the system increases and the operation latency decreases. This confirms our optimization for the procedure to implement a linearizable RSM. As the reads ratio increases, the writes ratio decreases. Note that in a lattice agreement instance the input lattice is formed only by all the writes. When the number of writes is small, the proposal command set would be small and the message size would be small as well. Thus, the system can complete a lattice agreement instance faster. This shows that the performance LaRSM is even better for settings with fewer writes. 

\begin{figure}[htb]
\includegraphics[width=3.5in]{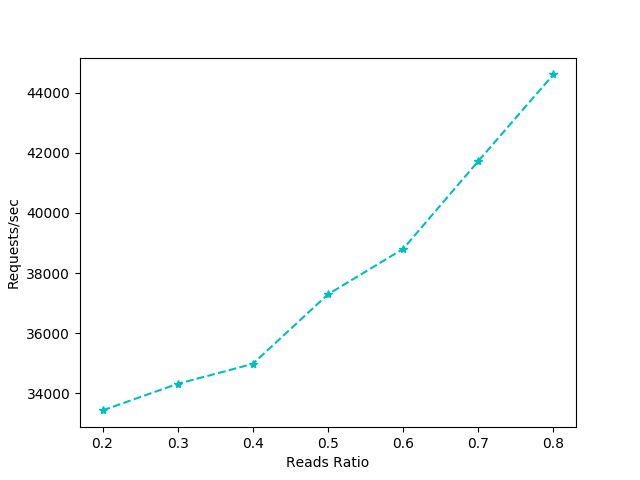}
\centering
\caption{Throughput under different reads ratio\label{fig:loads_throught}}
\end{figure}

\begin{figure}[htb]
\includegraphics[width=3.5in]{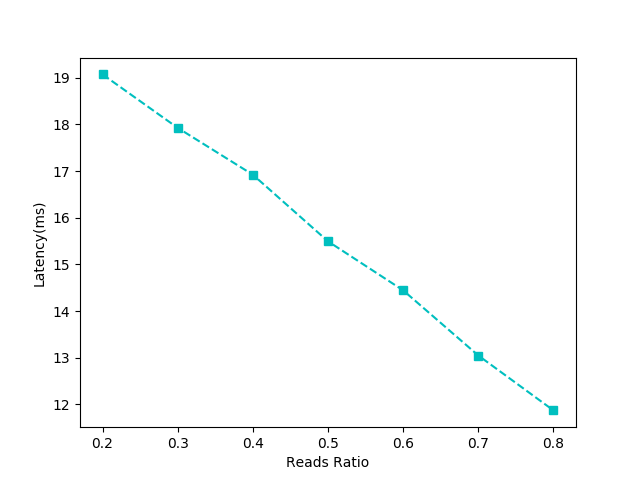}
\centering
\caption{Latency under different reads ratio\label{fig:loads_latency}}
\end{figure}

\subsection{Scalability Issue} 
Although LaRSM achieves good performance when the number of replicas in the system is small, its performance degenerates when the number of replicas increases, i.e, it is not scalable. The bad scalability is due to the fact that the lattice agreement protocol requires number of rounds that depends on the maximum number of crash failures the system can tolerate, which is typically set to be $\frac{n - 1}{2}$. In this case, as the number of replicas increases, the lattice agreement requires more rounds to complete. Therefore, LaRSM does not scale well. 

\section{Conclusion}
In this paper, we first give an algorithm to solve the lattice agreement problem in $O(\log f)$ rounds asynchronous rounds, which is an exponential improvement compared to previous $O(f)$ upper bound. This result also indicates that lattice agreement is a much weaker problem than consensus. In the second part, we explore the application of lattice agreement to building linearizable RSM. We first give improvements for the generalized lattice agreement protocol proposed in previous work to make it practical to implement a linearizable RSM. Then we perform experiments to show the effectiveness of our proposal. Evaluation results show that using lattice agreement to build a linearizable RSM has better performance than conventional consensus based RSM technique. Specifically, our implementation yields around 1.3x times throughput than SPaxos and incurs smaller latency, in normal case. In the failure case, LaRSM still continues to handle requests from clients, although its throughput decreases by some amount, whereas, SPaxos based protocol stops handling requests during the leader failure. 



\bibliography{ref}
\end{document}